\documentclass[letterpaper,12pt]{amsart}

\usepackage{latexsym, amsmath, amssymb,amsthm, natbib,verbatim, xy}
\xyoption{all}

\usepackage{tgpagella}
\usepackage{mathpazo}

\usepackage[left=.95in,top=1in,right=.95in,bottom=1in]{geometry}
\usepackage{setspace,color}
\usepackage{graphicx}
\usepackage{enumerate}
\usepackage[multiple]{footmisc}
\usepackage[leqno]{amsmath}
\usepackage{hyperref}
\usepackage{mathtools}

\usepackage{pgf,tikz}
\usepackage{graphicx}
\usetikzlibrary{arrows}
\definecolor{burgundy}{rgb}{0.5, 0.0, 0.13}
\definecolor{mediumtealblue}{rgb}{0.0, 0.33, 0.71}
\definecolor{blue(munsell)}{rgb}{0.0, 0.5, 0.69}
\definecolor{unitednationsblue}{rgb}{0.36, 0.57, 0.9}
\definecolor{mediumelectricblue}{rgb}{0.01, 0.31, 0.59}
\hypersetup{colorlinks,
            linkcolor=mediumtealblue,
            anchorcolor=mediumtealblue,
            citecolor=burgundy,urlcolor={mediumelectricblue}}

\usepackage{etoolbox}
\patchcmd{\section}{\scshape}{\bfseries}{}{}
\makeatletter
\renewcommand{\@secnumfont}{\bfseries}
\makeatother

\renewenvironment{quote}{%
   \list{}{%
     \leftmargin0.75cm   
     \rightmargin0.5cm
   }
   \item\relax
}
{\endlist}

\usepackage{tikz}

\newtheorem{theorem}{Theorem}
\newtheorem{corollary}{Corollary}

\newtheorem{proposition}{Proposition}

\usepackage{caption} 
\usepackage{algorithm, algorithmic}

\theoremstyle{definition}

\newtheorem{example}{Example}

\renewenvironment{quote}{%
   \list{}{%
     \leftmargin0.75cm   
     \rightmargin0.5cm
   }
   \item\relax
}
{\endlist}

\begin{document}

\title{Optimistic and pessimistic approaches for cooperative games}

\author[Atay and Trudeau]{Ata Atay \and Christian Trudeau}\thanks{\hspace{-0.75cm}\textbf{Date: }{\today.}\\ 
\textbf{Atay:} Department of Mathematical Economics, Finance and Actuarial Sciences, and Barcelona Economic Analysis Team (BEAT), University of Barcelona, Spain. E-mail: \href{mailto:aatay@ub.edu}{aatay@ub.edu}. \\
\textbf{Trudeau:} Department of Economics, University of Windsor, Windsor, ON, Canada. E-mail: \href{mailto:trudeauc@uwindsor.ca}{trudeauc@uwindsor.ca}.\\ \textbf{Acknowledgements:} Ata Atay is a Serra H\'{u}nter Fellow. Ata Atay gratefully acknowledges the support from the Spanish Ministerio de Ciencia e Innovaci\'{o}n research grant PID2023-150472NB-100/AEI/10.130339/501100011033, from the Generalitat de Catalunya research grant 2021-SGR-00306. Christian Trudeau gratefully acknowledges financial support by the Social Sciences and Humanities Research Council of Canada [grant number 435-2019-0141]. This material is based upon work supported by National Science Foundation under Grant No, DMS-1928930 while Ata Atay was in residence at the Mathematical Science Research Institute in Berkeley, California, during the Fall 2023 semester. We thank Mikel \'{A}lvarez-Mozos, Sreoshi Banerjee, Gustavo Bergantiños, Tuomas Sandholm, Vikram Manjunath, Leticia Lorenzo, Leanne Streekstra, María Gómez-Rúa, William Thomson, and the participants of the seminar at University of Vigo, 2023 Ottawa Microeconomic Theory Workshop, 2024 SAET conference (Santiago de Chile), 2024 SSCW Meeting (Paris), 2024 Spanish Social Choice Meeting.}

\begin{abstract}
Cooperative game theory studies how to allocate the joint value generated by a set of players. These games are typically analyzed using the characteristic function form with transferable utility, which represents the value attainable by each coalition. In the presence of externalities, coalition values can be defined through various approaches, notably by trying to determine the best and worst-case scenarios. Typically, the optimistic and pessimistic perspectives offer valuable insights into strategic interactions. In many applications, these approaches correspond to the coalition either choosing first or choosing after the complement coalition. In a general framework in which the actions of a group affects the set of feasible actions for others, we explore this relationship and show that it always holds in the presence of negative externalities, but only partly with positive externalities.  We then show that if choosing first/last corresponds to these extreme values, we also obtain a useful inclusion result: allocations that do not allocate more than the optimistic upper bounds also do not allocate less than the pessimistic lower bounds. Moreover, we show that when externalities are negative, it is always possible to guarantee the non-emptiness of these sets of allocations. Finally, we explore applications to illustrate how our findings provide new results and offer a means to derive results from the existing literature.
\end{abstract}
\maketitle

\newpage

\section{Introduction}
\label{sec:intro}

Cooperative game theory primarily focuses on groups of players
collaborating and pooling their profits/costs. A central question in this field is how the
profits (or costs) of a joint effort can be divided among the members of these coalitions.
A cooperative game is a characteristic function that specifies the value obtainable by each coalition of players. We focus on games with transferable utility (TU), which assume that each coalition can distribute its value in any way among its members. \footnote{See \cite{ps07} for a comprehensive introduction to cooperative games.}

In the presence of externalities it is not straightforward to determine the value that should be credited to a coalition, as the behavior of outside agents influences this value. In many applications, from queueing problems \citep{chun16} to minimum cost spanning tree problems \citep{bird1976} and river sharing problems \citep{AMBEC2002453}, various ways to define the value of a coalition have been proposed, in manners specific to the application at hand. These value functions are usually labeled as "optimistic" and "pessimistic", and typically correspond to the coalition in question acting before or after the complement set of agents.

We attempt to provide general answers on how to construct these coalitional functions in a general model containing what we call ``feasibility externalities". In this model, the revenue of an agent only depends on its own actions, that must be chosen from a feasible set. However, this feasible set might be impacted by the  actions taken by other agents. If the actions of others reduce the feasible set (compared to them being inactive), we say that we have negative feasibility externalities. Conversely, if the feasible set enlarges with the actions of others, we have positive feasibility externalities. The model includes as special cases the applications described above, as well as claims problems \citep{oneill82}, airport problems \citep{lo73} and joint production problems \citep{MOULIN1990, ms92}, among others.

To determine the value of a coalition we need to make assumptions on the behavior of the complement set of agents. We stand in opposition to the classic definitions of the $\alpha$ and $\beta$ values (\citealp{s82}, pp. 136--138), which either forces the coalition to maximize its minimum value, or sees the complement set minimize its maximum value, as both of these approaches sees the complement set of agents actively as trying to hurt the coalition. Our approach is closer to subgame perfection in that we rule out non-credible threats and suppose that the complement set of agents simply maximizes their (joint) benefits.

Our first question is as follows: do the optimistic and pessimistic value functions correspond to the value functions in which a coalition chooses their actions first (before the complement set of agents) or last (after them)?

Our results on this question are as follows:

\begin{enumerate}
    \item With negative feasibility externalities, the optimistic value function corresponds to a coalition choosing first, and the pessimistic value function corresponds to a coalition choosing last.
    \item With positive feasibility externalities, the pessimistic value function corresponds to a coalition choosing first, but the optimistic value function does not always correspond to a coalition choosing last.
\end{enumerate}

Intuitively, with negative externalities, choosing first gives a coalition the largest feasible set, leading to the optimistic value. Conversely, with positive externalities, choosing first results in the smallest feasible set, resulting in the pessimistic value. However, the implications of choosing last are not as clear. While we can compare the feasible set to when the complement set is inactive, the coalition might have alternate strategies to induce a more favorable set of actions from the complement set of agents.

For example, consider a situation in which a group of agents $N$ share a joint production facility for a common good, with the facility exhibiting increasing returns to scale. If a coalition $S$ chooses last, it might be that $N\setminus S$, faced with high marginal costs of production as the first users, only consumed a small amount of goods, leaving $S$ with relatively high marginal costs. It might be in the best interest of $S$ to have a subset $T$ of its members choose first, reducing the marginal costs for $N\setminus S$ and thus increasing their consumption, which in turns reduces the marginal costs faced by $S\setminus T$. It turns out that while these strategies can generate higher marginal contributions with positive externalities, they cannot yield lower marginal contributions with negative externalities.

The pessimistic value function provides lower bounds on what each coalition should receive, and thus it is natural to study its \emph{core} (\citealp{g59}), which is the set of allocations that distribute the total value while ensuring that each coalition receives at least its intrinsic value. On the other hand, the optimistic value function provides upper bounds on what each coalition should receive, and we should seek to prevent any coalition from surpassing its upper bound. The corresponding concept is the \emph{anti-core}, constructed by inverting the core inequalities (see, for instance, \citealp{oishi2016}). \footnote{This concept is rooted in considerations of fairness (\citealp{vew23}) but also serves as a measure of stability, as any coalition exceeding its best-case scenario may cause other agents to withdraw from cooperation.}

It is not immediately evident whether there exists a connection between the objective to ensure the pessimistic lower bounds and the objective of preventing anyone from exceeding the optimistic upper bounds. Our results demonstrate that it depends on the nature of externalities. Specifically, we show that there is a connection when optimistic/pessimistic corresponds to choosing first/last.

More precisely, we have that: 

\begin{enumerate} \item \textbf{Negative externalities}
\begin{enumerate}[(i)]
\item The anti-core of the optimistic value function is always a subset of the core of the pessimistic value function. That is, the objective of providing allocations that do not exceed the optimistic upper bounds is always at least as challenging as the objective of securing allocations surpassing the pessimistic lower bounds. 
\item Moreover, both sets are guaranteed to be non-empty, ensuring feasible solutions that provide no more than the optimistic upper bounds and no less than the pessimistic lower bounds.
\end{enumerate}
\item \textbf{Positive externalities}
\begin{enumerate}[(i)]
\item The relationship obtained in 1.(i) is guaranteed only if, for all coalitions, the optimistic value corresponds to its members acting last. 
\item The non-emptiness result in 1.(ii) is not guaranteed with positive feasibility externalities.
\end{enumerate}
\end{enumerate}

Altogether, our results are useful in multiple ways and they have significant implications on the analysis of many applications. For problems with negative feasibility externalities, we obtain the non-emptiness of the core across all cases, eliminating the need for case-by-case verification as in the previous literature. The relationship between the anti-core of the optimistic game and the core of the pessimistic game is obtained as soon as optimistic/pessimistic is equivalent to first/last; the anti-core of the optimistic game is then a refinement of the core of the pessimistic game. Given these relationships and the relative ease of computation, the results provide an argument to use first/last as reasonable approximations for lower/upper bounds even when the optimistic/pessimistic functions are not obtained from the first/last functions. 

We apply our results to various well-studied applications. These applications build on the links between TU games and joint optimization problems (see, among others, \citealp{kz82a}; \citealp{kz82b}; \citealp{gg92}). Applying our results to queueing,  minimum cost spanning tree and river sharing problems, we reobtain and reinterpret many classic results, while new results also emerge.

In some other applications, optimistic and pessimistic approaches yield dual games. We provide a sufficient condition for this to occur when optimistic/pessimistic corresponds to choosing first/last: the sequential and selfish decisions of a coalition picking first and its complement picking last must always yield an optimal outcome for the grand coalition. For instance, reconsider the joint production problem as above, with increasing returns to scale. If coalitions can adjust their consumption, both $S$ (choosing first) and $N\setminus S$ (choosing last) might reduce their consumption as they are faced with higher marginal costs and fail to internalize the externalities. The pessimistic and optimistic approaches are not dual, as the resulting consumption is sub-optimal. But if the demands are inelastic (as in \citealp{ms92} for instance), then we obtain the efficient amount of goods consumed, and the two approaches become dual.  

Hence, we establish the coincidence between the core of the pessimistic game and the anti-core of the optimistic game when the games defined on the order of coalitions arrival are dual. We illustrate duality by means of two well-known applications, bankruptcy (claims) and airport problems (see \cite{oneill82} and \cite{lo73} respectively).

\subsection{Related literature}
The list of methods proposed to determine the value of a coalition is long, starting with the $\alpha$ and $\beta$ games already discussed. 

The levels of cooperation within the coalition itself has been modeled in multi-choice (\citealp{hr93}) and fuzzy cooperative games (\citealp{a81}).

The question of how the complement coalition $N\setminus S$ would reorganize itself if a coalition $S$ breaks from the grand coalition is studied in partition function form games (\citealp{k18}). By opposition, we suppose that $N\setminus S$ still acts jointly, in a manner consistent with the behavior of $S$, an approach closer to \cite{hs03}.

Optimistic and pessimistic assumptions on the behavior of the complement set are common, for example in the coalition formation literature. Our idea that the complement set is using a strategy that is optimal for them, instead of trying to hurt $S$, is also found in \cite{rv97}. The non-cooperative interplay between $S$ and $N\setminus S$ is found notably in \cite{i81}. This strand of literature seeks a solution concept that consistently addresses the strategic effects of externalities in the coalition formation process, whereas our aim is to provide a unified model for TU games based on a joint optimization problems.

Closer to our perspective, \cite{Curiel_Tijs_1991} introduced two operators, minimarg and maximarg, which determine each coalition's marginal contribution based on the worst and the best possible order of agents, respectively. The minimarg assigns the smallest marginal contribution, while the maximarg assigns the largest, embodying pessimistic and optimistic viewpoints, respectively. Iteratively applying these operators to a game leads to a convex and concave game in the limit, with these games being dual to each other. Our approach differs in that they build these operators from a given value game, while we consider the underlying problem of how to define the games themselves.\medskip\\

The paper is organized as follows. Section \ref{sec:prel} provides some preliminaries on TU games. Section \ref{sec:model} introduces the framework and then tries to determine the value of a coalition. We set our ground rules in trying to determine  optimistic and pessimistic values on the value of a coalition. In Section \ref{sec:main} we provide our main results: (i) a relationship between choosing first/last and optimistic/pessimistic values, (ii) an inclusion result between the set of allocations making sure that no coalition gets more than the optimistic upper bounds and the one guaranteeing the pessimistic lower bounds, if first/last corresponds to optimistic/pessimistic.  In Section \ref{sec:app} we apply our model to a wide range of applications that have been well-studied in the literature. Section \ref{sec:extension} discusses extensions of our model. Finally, Section \ref{sec:conc} concludes.

\section{Preliminaries}
\label{sec:prel}
A \emph{cooperative game with transferable utility} (or TU game) is defined by a pair $(N, v)$ where $N$ is the (finite) set of agents and $v$ is a value function that assigns the value $v(S)$ to each coalition $S\subseteq N$ with $v(\emptyset)=0$. The number $v(S)$ is the value of the coalition. Whenever no confusion may arise as to the set of players, we will identify a TU game $(N,v)$ with its value function $v$.

Given a game $v$, an allocation is a tuple $x\in \mathbb{R}^{N}$ representing players' respective allotment. The total payoff of a coalition $S$ is denoted by $x(S)=\sum_{i\in S}x_{i}$ with $x(\emptyset)=0$. An allocation is \emph{efficient} if $x(N)=v(N)$, \emph{individually rational} if $x(i)\ge v(\{i\})$ for all $i\in N$, and \emph{coalitionally rational} if $x(S)\ge v(S)$ for all $S\subseteq N$.

An allocation is said to be in the \emph{core} of $v$ if it is efficient and coalitionally rational. Then, the core of the game $v$ is the set of all such allocations: $\mathcal{C}(v)=\left\{x\in\mathbb{R}^{N}: x(S)\ge v(S) \text{ for all } S\subset N \text{ and } x(N)=v(N)\right\}$. An allocation is said to be in the \emph{anti-core} of $v$ if it is efficient and for all coalitions the reversed coalitional rationality inequalities hold. Then, the anti-core of the game $v$ is the set of all such allocations: $\mathcal{A}(v)=\left\{x\in\mathbb{R}^{N}: x(S)\le v(S) \text{ for all } S\subset N \text{ and } x(N)=v(N)\right\}$. 

Let $\lambda: 2^N\setminus\{\emptyset\} \rightarrow [0,1]$ where for all $i\in N$ we have $\sum_{S\subseteq N: S\ni i}\lambda_S=1$. Let $\Lambda$ be the set of such balanded weights $\lambda$. A game is \emph{balanced} if $\sum_{S\subseteq N}\lambda_Sv(S)\leq v(N)$ for all $\lambda\in \Lambda$. A game $v$ has a non-empty core if and only if it is balanced \citep{bondareva1963,shapley1967}. A game is \emph{anti-balanced} if $\sum_{S\subseteq N}\lambda_Sv(S)\geq v(N)$ for all $\lambda\in \Lambda$. A game $v$ has a non-empty anti-core if and only if it is anti-balanced. 

Convexity and concavity (\citealp{s71}) are conditions that have been extensively studied to prove balancedness. A game $(N,v)$ is said to be \emph{convex} if $v(T\cup\{i\})-v(T)\ge  v(S\cup \{i\})-v(S)$ for all $i\in N$ and $S\subseteq T\subseteq N\setminus \{i\}$. A game $(N,v)$ is said to be \emph{concave} if $v(T\cup\{i\})-v(T)\le  v(S\cup \{i\})-v(S)$ for all $i\in N$ and $S\subseteq T\subseteq N\setminus \{i\}$.

The Shapley value (\citealp{s53}) is a single-valued solution that has interesting fairness properties. It is the weighted sum of the agents' marginal contributions to all coalitions. Formally, given a game $(N,v)$, the Shapley value $Sh(v)$ assigns to each agent $i\in N$ the payoff $Sh_{i}(v)=\sum\limits_{S\subseteq N\setminus \{i\}}\frac{|S|!(|N|-|S|-1)!}{|N|!}\left[v(S\cup\{i\})-v(S)\right]$.

A game $(N,v^{*})$ is the \emph{dual game} of the game $(N,v)$ if $v^{*}(S)=v(N)-v(N\setminus S)$ for all $S\subseteq N$.

For dual games, it is well-known that the anti-core of $v$ coincides with the core of $v^{*}$, and vice versa.
\begin{proposition}
\label{prop:dual}
If $v$ and $v^{*}$ are dual, then $\mathcal{A}(v)=\mathcal{C}(v^{*})$ and $\mathcal{A}(v^{*})=\mathcal{C}(v)$.
\end{proposition}

\section{The Model}
\label{sec:model}
We build a general model that allows for what we call \textit{feasibility externalities}, where the actions of others do not have a direct impact on the revenues one receives, but these actions might affect the set of actions one can take.

Formally, each agent $i\in N$ can take actions, with the set of possible actions defined as $\mathbb{A}_{i}.$ For each agent, the null action $\ominus_i \in \mathbb{A}_{i}$ means that one possible action is to stay inactive. For each $S\subseteq N,$ we define as $\mathbb{A}^{S}=\bigtimes\limits_{i\in S}\mathbb{A}_{i}$ the sets of actions jointly available to $S$ and $\mathbb{A}\equiv \mathbb{A}^N$.

When agents choose their actions, some actions might not be available. We thus define the \emph{feasible set}, which depends on the actions of other agents. More precisely, for all $S\subseteq N$ and all $a_{N\setminus S}\in \mathbb{A}^{N\setminus S},$ $f_{S}(a_{N\setminus S})\subseteq \mathbb{A}^{S}$ represents the set of actions jointly feasible for $S.$ We suppose that these sets are always non-empty, since for any coalition, all agents being inactive, $\ominus_S$, is always available as an action. Since the coalition $N$ includes all players, we write $f_N$ instead of $f_{N}(\ominus_{\emptyset})$. Let $f$ represent the set of all such feasibility functions for all coalitions $S$. We impose a \textbf{feasibility complementarity} condition: for all $S\subset N$ and $a_{N\setminus S} \in \mathbb{A}^{N\setminus S}$, $a_S\in f_S(a_{N\setminus S})$ if and only if $\left(a_S,a_{N\setminus S}\right)\in f_N$. In words, we assume that if a coalition selects first and the remaining agents select next, the combination of actions is jointly feasible for the grand coalition. Inversely, a set of feasible actions for the grand coalition must be such that if $N\setminus S$ picks their actions in that set first, the remaining actions are feasible for $S$. The condition is mild and satisfied by all the applications in this paper. Consider the following example that fails the condition: Both $S$ and $N\setminus S$ have a fixed budget to spend, but each dollar spent by $N\setminus S$ decreases the budget for $S$, while spending by $S$ has no impact on $N\setminus S$. Suppose that $S$ picks first and spends all of its budget, then $N\setminus S$ does the same. Then, the combination of their actions is not feasible for the grand coalition.

For each agent $i\in N$ we have a revenue function $R_{i}:\mathbb{A}_{i}\rightarrow \mathbb{R}$. Let $R$ represent the set of individual revenue functions. Given that we often have coalitions maximizing their joint revenues, if coalition $S$ chooses the set of actions $a_S$, we abuse notation and write $R_i(a_S)$ instead of $R_i((a_S)_i)$ for all $i\in S$.

The grand coalition faces an optimization problem that we generally write as $\max_{a_{N}\in f_{N}}\sum_{i\in N}R_{i}(a_{N})$. We define a problem $P$ as $(\mathbb{A},f,R),$ which describes the set of actions, the feasibility sets, and the revenue functions. We suppose that problem $P=(\mathbb{A},f,R)$ has a solution.
Let $\mathcal{P}$ be the set of all such problems (for all $\mathbb{A},f,R$).

\begin{example}
Suppose a simple queueing problem. All agents in $N$ have one single job to be processed on a machine. The machine can process one job per period, and agents have linear waiting costs: if agent $i$'s job is processed in period $t$, he suffers a cost of $t\times w_{i}$, where $w_i\geq 0$ is his personal waiting cost parameter. 

In this context, we can set $\mathbb{A}_i=\{1,\ldots,|N|\}$ to be the set of periods in which $i$'s job could be processed.\footnote{More precisely, $a_i\in \mathbb{A}_i$ means that the jobs starts processing in period $a_i-1$ and is completely processed at time $a_i$.} Then, for any $S\subseteq N$, $f_S(\ominus_{N\setminus S})$ represents what is jointly feasible for $S$ if $N\setminus S$ is inactive, i.e., if their jobs are not processed. We then have that $f_S(\ominus_{N\setminus S})$ is a function $\theta^S: S\rightarrow \mathbb{A}^S$ such that $\theta^S_i\neq \theta^S_j$ for all $i,j$ different in $S$. In words, no two agents in $S$ can be assigned the same processing period.

For $a_{N\setminus S}\neq \ominus_{N\setminus S}$, we have the additional constraint that $\theta^S_i\neq a_j$ for all $i\in S$ and $j\in N\setminus S$. Stated otherwise, the agents in $S$ cannot be assigned to a period already occupied by an agent in $N\setminus S$.

Finally, we have $R_i(a_i)=-w_ia_i$ for all $i\in N$, i.e., each agent has a disutility $w_i$ per period waiting before the job is processed.\footnote{If $a_i=\ominus$, then $R_i=-\infty$ as the agent's job is not processed.} We can thus rewrite the problem of the grand coalition as $\max_{\theta \in\Theta(N)}\sum_{i\in N} -\theta_i w_i$ where $\Theta(N)$ is the set of bijections from $N$ to $\{1,\ldots,|N|\}$.
\end{example}

\begin{example}
    We consider agents that share a joint production technology for a homogeneous good. The production technology is represented by a non-decreasing function $C:\mathbb{R}_+\rightarrow \mathbb{R}_+$ that assigns a cost to any quantity of good produced. Each agent must decide how much he wants to consume and how much to pay.
    
    In this context we can set $\mathbb{A}_i=(q_i,t_i)$ where $q_i\in \mathbb{R}_+$ is the amount consumed and $t_i\in \mathbb{R}$ is the money paid. We write $\ominus_i=(0,0)$.

    Then, for any $S\subseteq N$, $f_S(\ominus_{N\setminus S})$ represents what is jointly feasible for $S$ if $N\setminus S$ is inactive, i.e. if they do not consume any good and pay anything. 
    We then have that $f_S(\ominus_{N\setminus S})=\{(q_S,t_S)\in (\mathbb{R}_+^S,\mathbb{R}^S) \mid\sum_{i\in S}t_i\geq C\left(\sum_{i\in S}q_i \right)\}$. In words, $f_S$ is a budget set in which the coalition must collect enough money to cover for cost of the units it wants to consume.
    
    For $a_{N\setminus S}=(q_{N\setminus S},t_{N\setminus S})\neq \ominus_{N\setminus S}$, we have that 
    \[
    f_S(q_{N\setminus S},t_{N\setminus S})=\left\{(q_S,t_S)\in (\mathbb{R}^S_+,\mathbb{R}^S) \mid\sum_{i\in S}t_i+\sum_{i\in N\setminus S}t_i\geq C\left(\sum_{i\in S}q_i +\sum_{i\in N\setminus S}q_i\right)\right\}.
    \]
    In words, the payments of agents in $S$ must be enough to cover for the cost of the total number of units consumed, net of what was paid by $N\setminus S$. Note that if $N\setminus S$ just covered the cost of their consumption, then we need $\sum_{i\in S}t_i\geq C\left(\sum_{i\in S}q_i +\sum_{i\in N\setminus S}q_i\right)-C\left(\sum_{i\in N\setminus S}q_i\right)$.

    Finally, we have $R_i(q_i,t_i)=u_i(q_i)-t_i$ where $u_i:\mathbb{R}_+\rightarrow\mathbb{R}$ is a non-decreasing utility function.
\end{example}

\subsection{Externalities}

We say that a problem exhibits \textbf{negative externalities} if for all $i\in S\subseteq N$ and all $a_{N\setminus S}\in f_{N\setminus S}(\ominus_S)$, we have $f_{S}(a_{N\setminus S})\subseteq f_{S}(\ominus _{N\setminus S})$. Let $\mathcal{P}^{-}$ be the set of all such problems.

We say that a problem exhibits \textbf{positive externalities} if for all $i\in S\subseteq N$ and all $a_{N\setminus S}\in f_{N\setminus S}(\ominus_S)$,  we have $f_{S}(a_{N\setminus S})\supseteq f_{S}(\ominus _{N\setminus S})$. Let $ \mathcal{P}^{+}$ be the set of all such problems.




\subsection{Defining cooperative games}

\label{subsec:games}

It is not always trivial to determine what value to assign to a coalition.
In the presence of externalities, the value depends on assumptions we make
about the behavior of agents external to the coalition considered. 

A first approach consists in applying  the classic definitions of $\alpha$ and $\beta$ games. As a reminder, $\alpha$ games are maximin: the complement $S$ first tries to minimize the payoff to $S$, and then $S$ acts to maximize its payoff. By opposition $\beta$ games are minimax: $S$ tries to maximize its payoff, and then the complement $N\setminus S$ tries to minimize their payoffs. In our context, they are defined  as follows:
\begin{equation*}
v^{\alpha }(S)=\max_{a_{S}\in f_{S}(a_{N\setminus S})}\min_{a_{N\setminus
S}\in f_{N\setminus S}\left( \ominus _{S}\right) }\sum_{i\in S}R_{i}(a_{S})
\end{equation*}%
and 
\begin{eqnarray*}
v^{\beta }(S) &=&\min_{a_{N\setminus S}\in f_{N\setminus S}\left(
a_{S}\right) }\max_{a_{S}\in f_{S}(\ominus _{N\setminus S})}\sum_{i\in
S}R_{i}(a_{S}) \\
&=&\max_{a_{S}\in f_{S}(\ominus _{N\setminus S})}\sum_{i\in S}R_{i}(a_{S})
\end{eqnarray*}%
where the last simplification is because of our assumption that there is no
direct externalities. Thus, in $v^{\beta },$ $N\setminus S$ cannot do
anything. However, in $v^{\alpha },$ $N\setminus S$ can reduce the value of 
$S$ if its actions allows to reduce the feasible set of coalition $S.$

We want to compare these classic game definitions to some that have been
used in many applied problems, as queueing and minimum cost spanning tree problems, consisting in coalition $S$ choosing either first or
after $N\setminus S.$ If it chooses first, then we have that 
\begin{eqnarray*}
v^{F}(S) &=&\max_{a_{S}\in f_{S}(\ominus _{N\setminus S})}\sum_{i\in
S}R_{i}(a_{S}) \\
&=&v^{\beta }(S).
\end{eqnarray*}

If $S$ chooses after $N\setminus S,$ we specify that the objective
of $N\setminus S$ is not to harm $S,$ but rather to maximize its own revenues. In general, let $\mu _{S}(a_{N\setminus
S})=\arg \max_{a_{S}\in f_{S}(a_{N\setminus S})}\sum_{i\in S}R_{i}(a_{S})$ be the set of maximizers when $S$ is choosing after $N\setminus S$ has chosen $a_{N\setminus S}$. Since the feasible set --and thus the payoff $S$ might receive after $N\setminus S$ has chosen-- depends on which maximizer $N\setminus S$ has picked, we define minimum and maximum values as follows:

\begin{equation*}
v_{\min }^{L}(S)=\max_{a_{S}\in f_{S}(a_{N\setminus
S})}\min_{a_{N\setminus S}\in \mu _{N\setminus S}\left( \ominus
_{S}\right) }\sum_{i\in S}R_{i}(a_{S})
\end{equation*}
and 
\begin{equation*}
v_{\max }^{L}(S)=\max_{a_{S}\in f_{S}(a_{N\setminus
S})}\max_{a_{N\setminus S}\in \mu _{N\setminus S}\left( \ominus
_{S}\right) }\sum_{i\in S}R_{i}(a_{S}).
\end{equation*}

Notice that the definitions of $v^{\alpha }$ and $v_{\min }^{L}$ are very similar, with the only distinction being that we suppose for $v_{\min }^{L}$ that $N\setminus S$ can only select from its set of maximizers, while in $v^{\alpha }$ it can select any actions from its feasible set. Thus, the distinction is similar to restricting to credible threats. $v_{\max }^{L}$ is more optimistic in that we suppose that $N\setminus S$ selects its maximizer that is most favorable to $S$.

\begin{example}\label{ex_drts}
We consider a joint production problem as in the previous example, with the added simplification that agents consume the good in discrete units. We can then describe utility functions by vectors of marginal utilities and the cost function by a vector of marginal costs. 

Suppose that $N=\left\{ 1,2,3\right\} $ and that the
common technology of production exhibits decreasing returns to scale.

More precisely, agent 1 has marginal utilities of 6 for the first unit, 3
for the second, and zero afterwards. Agent 2 has marginal utility of 12 for
the first unit, 6 for the second, and zero afterwards. Agent 3 has marginal
utility of 12 for the first unit, 8 for the second, 4 for the third, and
zero afterwards. The marginal cost of producing the $x^{th}$ unit is $x-1$.

Given that it's always optimal for a coalition to pay just enough to cover the cost of its consumption, we abuse notation slightly by writing the set of maximizers as the vectors of quantities, without specifying the cost paid. 

We obtain the following values for the games we have defined:

\begin{center}
\begin{tabular}{ccccc}
$S$ & $v^{\alpha }(S)$ & $v_{\min }^{L}(S)$ & $v_{\max }^{L}(S)$ & $%
v^{F}(S)=v^{\beta }(S)$ \\ 
$\left\{ 1\right\} $ & 0 & 1 & 2 & 8 \\ 
$\left\{ 2\right\} $ & 0 & 9 & 9 & 17 \\ 
$\left\{ 3\right\} $ & 0 & 11 & 13 & 21 \\ 
$\left\{ 1,2\right\} $ & 0 & 12 & 12 & 21 \\ 
$\left\{ 1,3\right\} $ & 0 & 17 & 17 & 24 \\ 
$\left\{ 2,3\right\} $ & 0 & 24 & 24 & 32 \\ 
$\left\{ 1,2,3\right\} $ & 34 & 34 & 34 & 34%
\end{tabular}
\end{center}

We explain some of these values in detail. First, consider $v^{F}(\left\{
2,3\right\} ).$ The coalition chooses first, and faces low marginal
costs. If it produces 5 units (2 for agent 2 and 3 for agent 3), it obtains $%
12-0+12-1+8-2+6-3+4-4=32.$ If it produces $4$ units (2 for agent 2 and 2 for
agent 3), it obtains $12-0+12-1+8-2+6-3=32.$ It is easy to see that any
other combination yields less net revenues. Thus, $\mu _{\left\{ 2,3\right\}
}\left( \ominus _{1}\right) =\left\{ (2,3),(2,2)\right\} $ and $%
v^{F}(\left\{ 2,3\right\} )=32.$

Now consider $v_{\min }^{L}(\left\{ 1\right\} ).$ We suppose that
coalition $\left\{ 2,3\right\} $ has selected an action that maximizes its own revenue. For $\{1\}$ the worst maximizer in $\mu _{\left\{ 2,3\right\} }\left( \ominus_{1}\right) $ is $(2,3),$ which means that $5$ units have already been produced, and that agent 1 now faces a marginal cost of 5 on the first unit consumed, 6 on the second, etc. Given that, it is optimal to consume a single unit for a net revenue of $6-5=1=v_{\min }^{L}(\left\{ 1\right\} ).$ Moving to $v_{\max }^{L}(\left\{ 1\right\} ),$ we now suppose that coalition $\left\{ 2,3\right\} $ has picked its maximizer that is most favorable to
agent 1, here $(2,2).$ Thus, agent 1 is now facing a marginal cost of 4 on
the first unit consumed, 5 on the second, etc. It is still optimal to
consume a single unit, but the net revenue is now $6-4=2=v_{\max}^{L}(\left\{ 1\right\} ).$

Finally, consider $v^{\alpha }(\left\{ 1\right\} ).$ How can coalition $%
\left\{ 2,3\right\} $ hurt the most agent 1 if it is not constrained by
choosing something that maximizes its own revenues? If it selects to consume
at least 6 units, then agent 1 faces a marginal cost of 6 of the first unit
consumed, and he is incapable of generating a strictly positive value, and
thus $v^{\alpha }(\left\{ 1\right\} )=0.$ Here, coalition $N\setminus S$
can always choose a quantity large enough to push $v^{\alpha }(S)=0$ for all 
$S\neq N$.
\end{example}

We can see that in the previous example $v^{\alpha }$ is not interesting because without the constraint of selecting a maximizer, we obtain values that are too low. In other words, $v^{\alpha}$ is too pessimistic as it considers the complement set taking non-credible actions.

We now formally link $v^{\alpha}$ and $v^{\beta}$ to the games where a coalition picks either first or last. The result is easily obtained and offered without proof.

\begin{proposition}
\mbox{}

\begin{itemize}
\item[(i)] For all $P\in \mathcal{P}^{-}$ and all $S\subseteq N$, we have $%
v^{\alpha }(S,P)\leq v_{min}^{L}(S,P)\leq v_{max}^{L}(S,P)\leq
v^{F}(S,P)=v^{\beta }(S,P)$.

\item[(ii)] For all $P\in \mathcal{P}^{+}$ and all $S\subseteq N$, we have $%
v^{\alpha }(S,P)=v^{\beta }(S,P)=v^{F}(S,P)\leq v_{min}^{L}(S,P)\leq
v_{max}^{L}(S,P)$. 
\end{itemize}
\end{proposition}

Notice that when we have positive externalities, the worst that $%
N\setminus S$ can do is to stay inactive, keeping the feasible set of $S$ as small as possible, and thus $v^{\alpha }(S,P)=v^{\beta }(S,P)=v^{F}(S,P).$

Given these results, we can see that $v^{\beta }$ is always equal to $v^{F}$, and that when we have positive externalities we also have that $v^{\alpha }$ is equal to $v^{F}.$ With negative externalities, $v^{\alpha }$ can differ from other games, but is singularly pessimistic in supposing that agents in $N\setminus S$ go out of their way to hurt $S$. Thus, $v^{\alpha}$ and $v^{\beta}$ do not seem like the right approaches in our set of problems, with first/last to choose being better approaches.

But, more generally, we are interested in defining lower and upper bounds on the
value a coalition can generate. We can see in the previous result that with negative externalities, choosing first is favorable scenario, and offers an upper bound, while choosing last is an unfavorable scenario that offers a lower bound. The ranking is flipped with positive externalities.

But, ideally, we'd like to consider the most-favorable and least-favorable scenarios. Do these align with choosing first or last?

Formally,  under the \emph{optimistic approach}, each coalition is assigned a value corresponding to a best-case scenario which we can interpret as an upper bound on the value coalition $S$ can achieve. Thus, the relevant concept to study is the anti-core under the
optimistic approach, as it is the set of efficient payoff vectors that
assigns to each coalition at most its value.

By opposition, in the \emph{pessimistic approach} each coalition is assigned a value corresponding to a
worst-case scenario, which we can interpret as a lower bound on the value coalition $S$ can achieve. Thus, the relevant concept to study is the core under the pessimistic approach, as it is the set of efficient payoff vectors that assigns to each coalition at least as much as its value.

We follow two principles when defining these coalitional values. First, we suppose that the value of coalition $S$ is defined through a maximization over its own actions, although this maximization might be constrained by the actions taken by $N\setminus S$. Without this assumption, for the pessimistic value, it might be possible to find a very low value associated to coalition $S$ by having coalition $S$ choose a particularly self-harmful set of actions. We rule out this form of self-sabotage.

Secondly, we suppose that if some outside agents have taken actions before $S$, they have taken actions that are optimal for them, as opposed to have taken them with the express goal of hurting or benefiting $S$. Therefore, our approach is similar to subgame perfect Nash equilibrium in
non-cooperative game theory, in which we suppose that everyone's actions is
optimal in every subgame, thus ruling out non-credible threats, and, in particular, $v^{\alpha }$. For instance, \cite{ab22} define a pessimistic value for knapsack problems by supposing that outsiders pick the worst possible combination of actions for $S$, even if this combination is not optimal for the outsiders. This is what is considered in our approach as a non-credible threat and ruled out by this second principle. In the same way, when defining the optimistic value with positive externalities, one could
suppose that outsiders take a particularly advantageous set of actions for $S$, but such non-credible gesture is ruled out by our second principle.

How do these optimistic and pessimistic games relate to $v^{F},$ $v_{\min}^{L}$ and $v_{\max }^{L}?$ In theory, to define our optimistic and pessimistic values, we could have the members of $S$ pick their actions in any possible order, some first, some after some members of $N\setminus S$ but before others, and some last. Our two principles help narrow the search. First, if we allow agents in $N\setminus S$ to take some actions, by our second principle, we cannot suppose that only some of them take actions. Supposing that only some of them take actions is akin to supposing a particularly hurtful/beneficial set of actions. Once we let agents in $N\setminus S$ take actions, they do so for their own benefit. While some of its members can choose the null action, $N\setminus S$ cannot delay its decision further. In addition, given the cooperative nature of the problem, we suppose that the actions taken by $N\setminus S$ are jointly optimal for them.  As in \cite{hs03}, each coalition acts ``together'' in order to maximize their joint revenues. Moreover, each coalition $S$ chooses its action knowing that $N\setminus S$ chooses its action consistent with the problem and tries to maximize its own joint revenues. 

Therefore, the problem for coalition $S$ consists in determining its subset $T\subseteq S$ that will pick first. Then, the outside agents in $N\setminus S$ will jointly choose their actions, and agents in $S\setminus T$ choose last. Of course, if $T=S$, all agents in $S$ choose first, and if $T=\emptyset$, all agents in $S$ choose last.


As when we defined $v^L_{min}$ and $v^L_{max}$, while we suppose that $N\setminus S$ is optimizing its joint
revenue under the constraint imposed by the choice of $T$, there might be multiple maximizers, and which one is chosen might affect how much $%
S\setminus T$ might obtain afterwards. We thus define two interesting
marginal contributions for $S$ if $T$ chooses first and $S\setminus T$
chooses last, depending if we suppose that the maximizer chosen by $%
N\setminus S$ is best or worst for $S\setminus T$. 
Then, for any $%
T\subseteq S\subseteq N$ we define:\footnote{
We suppose, in this problem and in subsequent ones, that the optimization
problem for coalition $S\subset N$ has a solution.}

\begin{equation*}
v^{T\subseteq S}_{max}(P)=\max_{a_T\in f_T\left(\ominus_{N\setminus
T}\right)}\left(\sum_{i\in T}R_i(a_T)+\max_{a_{N\setminus S}\in
\mu_{N\setminus S}\left(\{a_T,\ominus_{S\setminus
T}\}\right)}\max_{a_{S\setminus T}\in f_{S\setminus T}\left(a_T,
a_{N\setminus S}\right)}\sum_{i\in S\setminus T}R_i(a_{S\setminus T})\right) 
\end{equation*}

and

\begin{equation*}
v^{T\subseteq S}_{min}(P)=\max_{a_T\in f_T\left(\ominus_{N\setminus
T}\right)}\left(\sum_{i\in T}R_i(a_T)+\min_{a_{N\setminus S}\in
\mu_{N\setminus S}\left(\{a_T,\ominus_{S\setminus
T}\}\right)}\max_{a_{S\setminus T}\in f_{S\setminus T}\left(a_T,
a_{N\setminus S}\right)}\sum_{i\in S\setminus T}R_i(a_{S\setminus
T})\right). 
\end{equation*}

We make two remarks: First, there are three steps to this maximization. $T$
picks first, then its actions are seen by $N\setminus S$ who then maximizes its revenues, and finally $S\setminus T$ picks, having observed actions of
the first two groups. Second, notice that we still suppose joint
maximization for $S$, with $T\subseteq S$ anticipating the impact of its actions on $N\setminus S$ and the subsequent impact of these decisions on $S\setminus T$.

Naturally, we define the optimistic value of $S$ as its largest marginal
contribution as defined above, while the pessimistic value of $S$ is its
smallest marginal contribution. Formally, for all $S\subseteq N$
\begin{equation*}
v^o(S,P)=\max_{T\subseteq S} v^{T\subseteq S}_{max}(P) 
\end{equation*}
and 
\begin{equation*}
v^p(S,P)=\min_{T\subseteq S} v^{T\subseteq S}_{min}(P). 
\end{equation*}


It is easy to see that when $T=S$, we recover $v^{F}$: $v^{F}(P)\equiv
v_{max}^{S\subseteq S}(P)=v_{min}^{S\subseteq S}(P)$, as which maximizer is
chosen by $N\setminus S$ is irrelevant, as nobody chooses after them. When $%
T=\emptyset $, all agents in $S$ chooses last. The maximizer chosen by $%
N\setminus S$ now matters, and we recover $v_{max}^{L}(P)\equiv
v_{max}^{\emptyset \subseteq S}(P)$ and $v_{min}^{L}(P)\equiv
v_{min}^{\emptyset \subseteq S}(P)$.
\medskip\\

While choosing first/last yields natural bounds, are they always the lower and upper bounds? In the previous example, it is the case. The best case scenario for a coalition is to choose its consumption first, when the marginal costs are low. The worst case is choosing last, when the marginal costs are higher.

The following example, in which we flip from decreasing to increasing returns to scale, shows that it is not always the case. 

\begin{example}\label{ex_irts}

We modify the previous example to suppose increasing returns to scale in
production. Suppose the same marginal utilities, but now the marginal cost
of production is 14 for the first unit, 9 for the second, 7 for the third, 3
for the fourth and 1 afterwards.

We obtain the following values:
\begin{center}
\begin{tabular}{cccc}
$S$ & $v^{F}(S)$ & $v_{\min }^{L}(S)=v_{\max }^{L}(S)$ & $v^{o}(S)$ \\ 
$\left\{ 1\right\} $ & 0 & 7 & 7 \\ 
$\left\{ 2\right\} $ & 0 & 0 & 0 \\ 
$\left\{ 3\right\} $ & 0 & 0 & 0 \\ 
$\left\{ 1,2\right\} $ & 0 & 0 & 8 \\ 
$\left\{ 1,3\right\} $ & 0 & 0 & 11 \\ 
$\left\{ 2,3\right\} $ & 8 & 8 & 10 \\ 
$\left\{ 1,2,3\right\} $ & 15 & 15 & 15%
\end{tabular}
\end{center}

We explain these values for $\left\{ 1\right\}$ and $\left\{ 2,3\right\}$.
If agent 1 has to choose first, it faces too high marginal costs, and it
consumes nothing, and $v^{F}(\left\{ 1\right\} )=0.$ Coalition $\left\{
2,3\right\} $, acting first, will pick $(2,3),$ to generate a net surplus of 
$12-14+12-9+8-7+6-3+4-1=8=v^{F}(\left\{ 2,3\right\} ).$

Now consider $v_{\min }^{L}(\left\{ 1\right\} ).$ We have that coalition $%
\left\{ 2,3\right\} $ has picked $(2,3),$ so agent 1 is now facing marginal
costs of 1. He consumes 2 units for a gain of $6-1+3-1=7=v_{\min
}^{L}(\left\{ 1\right\} )=v_{\max }^{L}(\left\{ 1\right\} ).$ Since $\left\{
1\right\} ,$ when alone, doesn't consume, we have that $v_{\min
}^{L}(\left\{ 2,3\right\} )=v_{\max }^{L}(\left\{ 2,3\right\}
)=v^{F}(\left\{ 2,3\right\} )=8.$

But $\left\{ 2,3\right\} $ can do better. Consider $v_{\min }^{\left\{
3\right\} \subseteq \left\{ 2,3\right\} }.$ If $3$ picks first and consumes
three units, it generates $12-14+8-9+4-7=-6$ but it decreases the marginal
cost for agent 1, who then consumes 2 units, leaving a marginal cost of 1
for agent 2, who consumes 2, generating $12-1+6-1=16,$ and thus $v^{o}\left(
\left\{ 2,3\right\} \right) =-6+16=10$. 
\end{example}
The previous example shows that we do not have a direct correspondence between first/last and optimistic/pessimistic. But since choosing first/last is considerably simpler than looking for the
best/worst case scenarios more generally, in the next section we explore
when they correspond to the optimistic/pessimistic scenarios.

\section{Main results}

\label{sec:main}

We now provide our main results. First, we establish a relationship between the games defined according to pessimistic/optimistic approaches and
coalitions choosing first/last. 

\begin{proposition}
\label{prop:v} \mbox{}

\begin{itemize}
\item[(i)] For all $P\in \mathcal{P}^{-}$ and all $S\subseteq N$, we have $%
v^p(S, P) \leq v^L_{min}(S, P) \leq v^F(S, P) \leq v^o(S, P)$. For all $P\in 
\mathcal{P}^{-}$ we have $\mathcal{C}(v^L_{min}(\cdot, P))\subseteq \mathcal{%
C}(v^p(\cdot, P))$ and $\mathcal{A}(v^F(\cdot, P))\subseteq \mathcal{A}%
(v^o(\cdot, P))$. 

\item[(ii)] For all $P\in \mathcal{P}^{+}$ and all $S\subseteq N$, we have $%
v^p(S, P) \leq v^F(S, P) \leq v^L_{max}(S, P) \leq v^o(S, P)$. For all $P\in 
\mathcal{P}^{+}$ we have $\mathcal{C}(v^F(\cdot, P))\subseteq \mathcal{C}%
(v^p(\cdot, P))$ and $\mathcal{A}(v^L_{max}(\cdot, P))\subseteq \mathcal{A}%
(v^o(\cdot, P))$.
\end{itemize}
\end{proposition}

Since $v^p$ and $v^o$ look at all possibilities for the coalition, including picking first or picking last, the relationships with $v^F$ and $v^L$ are
obvious. The relationship between $v^F$ and $v^L$ depend on the
externalities. With negative externalities, the feasible set is larger when
picking first compared to last, while the opposite is true if we have
positive externalities.

It turns out that when we have externalities that are always of the same
sign, there is a close relationship between first/last and the
optimistic/pessimistic values.

\begin{theorem}
\label{v^T}
\begin{itemize}
\item[(i)] For all $P\in \mathcal{P^-}$, we have $v^{o}(S,P)=v^F(S,P)$ for
all $S\subset N$. 

\item[(ii)] For all $P\in \mathcal{P^-}$, we have $v^p(S,P)=v^L_{min}(S,P)$
for all $S\subset N$. 

\item[(iii)] For all $P\in \mathcal{P^+}$, we have $v^{p}(S,P)=v^F(S,P)$ for
all $S\subset N$.
\end{itemize}
\end{theorem}

\begin{proof}
(i) Suppose that $P\in \mathcal{P^-}$. We have that 
\begin{eqnarray*}
v^{T\subseteq S}_{max}(P)&=&\max_{a_T\in f_T\left(\ominus_{N\setminus
T}\right)}\left(\sum_{i\in T}R_i(a_T)+\max_{a_{N\setminus S}\in
\mu_{N\setminus S}\left(\{a_T,\ominus_{S\setminus
T}\}\right)}\max_{a_{S\setminus T}\in f_{S\setminus T}\left(a_T,
a_{N\setminus S}\right)}\sum_{i\in S\setminus T}R_i(a_{S\setminus T})\right)
\\
&\leq &\max_{a_T\in f_T\left(\ominus_{N\setminus T}\right)}\left(\sum_{i\in
T}R_i(a_T)+\max_{a_{S\setminus T}\in f_{S\setminus T}\left(a_T,
\ominus_{N\setminus S}\right)}\sum_{i\in S\setminus T}R_i(a_{S\setminus
T})\right) \\
&=&\max_{a_{S}\in f_{S}\left( \ominus _{N\setminus S}\right) }\sum_{i\in
S}R_{i}\left( a_{S}\right) \\
&=&v^{F}(S,P)
\end{eqnarray*}%
where the inequality is because $f_{S\setminus T}\left( \left\{
a_{T},a_{N\setminus S})\right\} \right) \subseteq f_{S\setminus T}\left(
\left\{ a_{T},\ominus _{N\setminus S}\right\} \right) $ for all $%
a_{N\setminus S}\in f_{N\setminus S}\left(\{a_T,\ominus_{S\setminus
T}\}\right)$ by negative externalities.

(ii) Suppose that $P\in \mathcal{P^-}$. Let $a^{T\subseteq S,-}_{N\setminus
S}(a_T)$ be among the maximizers for $N\setminus S$ choosing after $T$ that
is the worst for $S\setminus T$.

We have that
\begin{eqnarray*}
v^{T\subseteq S}_{min}(P)&=&\max_{a_T\in f_T\left(\ominus_{N\setminus
T}\right)}\left(\sum_{i\in T}R_i(a_T)+\min_{a_{N\setminus S}\in
\mu_{N\setminus S}\left(\{a_T,\ominus_{S\setminus
T}\}\right)}\max_{a_{S\setminus T}\in f_{S\setminus T}\left(a_T,
a_{N\setminus S}\right)}\sum_{i\in S\setminus T}R_i(a_{S\setminus T})\right)
\\
&=&\max_{a_T\in f_T\left(\{a^{T\subseteq S,-}_{N\setminus S}(a_T),
\ominus_{S\setminus T}\}\right)}\left(\sum_{i\in
T}R_i(a_T)+\max_{a_{S\setminus T}\in f_{S\setminus T}\left(a_T,
a^{T\subseteq S,-}_{N\setminus S}(a_T)\right)}\sum_{i\in S\setminus
T}R_i(a_{S\setminus T})\right) \\
&=&\max_{a_S\in f_S\left(a^{T\subseteq S,-}_{N\setminus
S}(a_T)\right)}\sum_{i\in S}R_i(a_S) \\
&\geq &\max_{a_{S}\in f_{S}\left( a^{T\subseteq S,-}_{N\setminus
S}(\ominus_T)\right) }\sum_{i\in S}R_{i}\left( a_{S}\right) \\
&= &\max_{a_{S}\in f_{S}\left( a^{\emptyset\subseteq S,-}_{N\setminus
S}\right)}\sum_{i\in S}R_{i}\left( a_{S}\right) \\
&=&\min_{a_{N\setminus S}\in\mu_{N\setminus S}(\ominus_S)}\max_{a_{S}\in
f_{S}\left( a_{N\setminus S}\right)}\sum_{i\in S}R_{i}\left( a_{S}\right) \\
&=&v^{L}_{min}(S,P)
\end{eqnarray*}
where the second equality follows from feasibility complementarity. To see why the inequality holds, let $a^L_S$ be the actions taken by $S$ if it chooses last, and $N\setminus S$ has chosen $a^{\emptyset\subseteq S,-}_{N\setminus S}$. On the left hand side of the inequality, $T$ could choose any action, while on the right hand side, it chooses $\ominus_T$. But by negative externalities, $f_{N\setminus
S}\left( \{a_T, \ominus_{S\setminus T}\}\right)\subseteq f_{N\setminus
S}\left( \ominus_{S}\right)$. In addition, if $T$ chooses $\left(
\left(a^L_S\right)_T\right)$, then, by feasibility complementarity, $a^{\emptyset\subseteq S,-}_{N\setminus
S}\in f_{N\setminus S}\left( \left(a^L_S\right)_T\right)$, and thus is still chosen by $N\setminus S$. In other
words, when $T$ can choose actions that might influence $N\setminus S$, it
can always induce it to choose what it would choose if it picked first. If $T
$ picks anything else, it must give at least as large revenues, which
concludes the proof.

(iii) Suppose that $P\in \mathcal{P^+}$.

We have that, for all $T\subseteq S$, 
\begin{eqnarray*}
v^{T\subseteq S}_{min}(P)&=&\max_{a_T\in f_T\left(\ominus_{N\setminus
T}\right)}\left(\sum_{i\in T}R_i(a_T)+\min_{a_{N\setminus S}\in
\mu_{N\setminus S}\left(\{a_T,\ominus_{S\setminus
T}\}\right)}\max_{a_{S\setminus T}\in f_{S\setminus T}\left(a_T,
a_{N\setminus S}\right)}\sum_{i\in S\setminus T}R_i(a_{S\setminus T})\right)
\\
&\geq &\max_{a_T\in f_T\left(\ominus_{N\setminus T}\right)}\left(\sum_{i\in
T}R_i(a_T)+\max_{a_{S\setminus T}\in f_{S\setminus T}\left(a_T,
\ominus_{N\setminus S}\right)}\sum_{i\in S\setminus T}R_i(a_{S\setminus
T})\right) \\
&=&\max_{a_{S}\in f_{S}\left( \ominus _{N\setminus S}\right) }\sum_{i\in
S}R_{i}\left( a_{S}\right) \\
&=&v^{F}(S,P)
\end{eqnarray*}%
where the inequality is because $f_{S\setminus T}\left( \left\{
a_{T},a_{N\setminus S})\right\} \right) \supseteq f_{S\setminus T}\left(
\left\{ a_{T},\ominus _{N\setminus S}\right\} \right) $ for all $%
a_{N\setminus S}\in f_{N\setminus S}\left(\{a_T,\ominus_{S\setminus
T}\}\right)$ by positive externalities.
\end{proof}

As seen in the previous example, we do have counterexamples where $%
v^{o}(S)>v_{\max }^{L}(S)$  under positive feasibility externalities, unlike negative feasibility externalities.

While this shows that for positive externalities, choosing last is not necessarily the most optimistic value, we still have many applications in which it is the case. As seen in the counterexample and in the proof, it is only when a coalition can alter in its favor the behavior of the complement set that it might have a better option than choosing last.

We next show another nice feature of the first/last games. If we use them as lower/upper bounds, then we immediately obtain an inclusion result: the set of allocations making sure that nobody receives more than their upper bounds is a subset of the set of allocations making sure that nobody receives less than their lower bounds.



\begin{theorem}
\label{main} For all $P\in \mathcal{P}$ 
we have that

\begin{enumerate}[(i)]

\item $\mathcal{A}\left( v^{F}(\cdot,P)\right) \subseteq \mathcal{C}\left(
v^{L}_{min}(\cdot,P)\right)$;

\item $\mathcal{A}\left( v^{L}_{max}(\cdot,P)\right) \subseteq \mathcal{C}%
\left(v^{F}(\cdot,P)\right)$; 
\end{enumerate}
\end{theorem}

\begin{proof}
We start with part (\textbf{i}). 
Fix $P,$ and thus write $v^{F}(S)$ and $v^{L}_{min}(S)$ instead of $v^{F}(S,P)$
and $v^{L}_{min}(S,P).$ Let $v(N)\equiv \max_{a_{N}\in f_{N}}\sum_{i\in
N}R_{i}(a^{N}).$ Notice that $v^{F}(N)=v^{L}_{min}(N)=v(N).$

An allocation $x\in \mathcal{A}\left( v^{F}\right) $ if $v(N)-v^{F}(N%
\setminus S)\leq x(S)\leq v^{F}(S)$ for all $S.$ An allocation $x\in 
\mathcal{C}\left( v^{L}_{min}\right) $ if $v^{L}_{min}(S)\leq x(S)\leq
v(N)-v^{L}_{min}(N\setminus S).$ It is easy to see that $\mathcal{A}\left(
v^{F}\right) \subseteq \mathcal{C}\left( v^{L}_{min}\right) $ if and only if 
$v(N)\geq v^{F}(S)+v^{L}_{min}(N\setminus S)$ for all $S.$ Fix $S$ and let ${%
a}_{N}^{\ast }$ be (one of) the optimal set(s) of actions taken by $N.$
Then, $v(N)=\sum\limits_{i\in N}R_{i}({a}_{N}^{\ast }),$ $%
v^{F}(S)=\sum\limits_{i\in S}R_{i}\left( a^-_{F(S)}\right)$ and $%
v^{L}_{min}(N\setminus S)=\sum\limits_{i\in N\setminus S}R_{i}\left(
a^{-}_{L(N\setminus S)} \right)$, where $a_{F(S)}^-$ is the worst
maximizer for $N\setminus S$ among the maximizers when $S$ chooses first and 
$a^{-}_{L(N\setminus S)}$ is a maximizer for $N\setminus S$ after $S$ has
chosen $a_{F(S)}^-$. We thus have that 
\begin{eqnarray*}
v^{F}(S)+v^{L}(N\setminus S) & =& \sum\limits_{i\in S}R_{i}\left( a^-_{F(S)}
\right)+\sum\limits_{i\in N\setminus S}R_{i}\left( a^-_{L(N\setminus S)}
\right) \\
&=&\sum\limits_{i\in N}R_{i}\left(\{ a^-_{F(S)},a^-_{L(N\setminus
S)}\}\right) \\
&\leq &\sum\limits_{i\in N}R_{i}({a}^*_{N}) \\
&=&v(N),
\end{eqnarray*}
where the inequality follows from the fact that, by feasibility
complementarity, $\{ a^-_{F(S)},a^-_{L(N\setminus S)}\} \in f_{N}.$

Next, we show (\textbf{ii}). An allocation $x\in \mathcal{A}\left(
v^{L}_{max}\right) $ if $v(N)-v^{L}_{max}(N\setminus S)\leq x(S)\leq
v^{L}_{max}(S)$ for all $S.$ An allocation $x\in \mathcal{C}%
\left(v^{F}\right) $ if $v^{F}(S)\leq x(S)\leq v(N)-v^{F}(N\setminus S).$ It
is easy to see that $\mathcal{A}\left( v^{L}_{max}\right) \subseteq \mathcal{%
C}\left( v^{F}\right) $ if and only if $v(N)\geq
v^{L}_{max}(S)+v^{F}(N\setminus S)$ for all $S$, which we have already shown in the proof of part (\textbf{i}). 
\end{proof}



When the first/last games are actually the best/worst cases, we can strengthen these inclusion results. 

\begin{corollary}[of Theorem \protect\ref{main}]
\label{cor:main} For all $P\in \mathcal{P}^{-}$ we have that $\mathcal{A}%
\left( v^{o}(\cdot,P)\right) \subseteq \mathcal{C}\left(v^{p}(\cdot,P)\right)
$. 

For all $P\in \mathcal{P}^{+}$ we have that $\mathcal{A}\left(
v^{L}_{max}(\cdot,P)\right) \subseteq \mathcal{C}\left(v^{p}(\cdot,P)\right)$.
\end{corollary}

\begin{proof}
If $P\in \mathcal{P}^{-},$ then $v^{o}=v^{F}$ and $v^{p}=v^{L}_{min}.$ The
result follows from the first part of Theorem \ref{main}. If $P\in \mathcal{P%
}^{+},$ then $v^{p}=v^{F}.$ It follows from the second part of Theorem \ref%
{main}.
\end{proof}


Next, we show that when we have negative externalities the anti-core of the optimistic game is always non-empty.

\begin{theorem}
\label{main_non_empty} For all $P\in \mathcal{P}^-$, $\mathcal{A}\left(
v^{o}(\cdot,P)\right)$ is non-empty.
\end{theorem}

\begin{proof}
Fix $P\in \mathcal{P}^-$ and write $v^F(S)$ and $v^o(S)$ instead of $v^F(S,P) $ and $v^o(S,P)$. Since $P\in \mathcal{P}^-$ we have that $v^o=v^F$.

Let $a^*$ be (one of) the maximizer(s) for the problem of the grand
coalition and $a_{F(S)}$ be one of the maximizers when $S$ selects first. We first show that $v^o(S)\geq \sum\limits_{i\in S}R_{i}\left( a^{\ast }\right)$.

We have that 
\begin{eqnarray*}
v^{o}(S) &=&\sum_{i\in S}R_{i}\left( a_{F(S)} \right) \\
&\geq &\sum_{i\in S}R_{i}\left( a_{S}^{\ast } \right) \\
&=&\sum_{i\in S}R_{i}\left( a^{\ast }\right),
\end{eqnarray*}
where the first inequality is by definition of $a_{F(S)}$, since $a^*_S\in f_S(\ominus_{N\setminus S})$.

Then, take $\lambda\in \Lambda$ and multiply the inequality by $\lambda _{S}$
and sum over $S$ to obtain 
\begin{eqnarray*}
\sum_{S\subset N}\lambda _{S}v^{o}(S) &=&\sum_{S\subset N}\lambda
_{S}\sum_{i\in S}R_{i}\left( a_{F(S)} \right) \\
&\geq &\sum_{S\subset N}\lambda _{S}\sum_{i\in S}R_{i}\left( a^{\ast }\right)
\\
&=&\sum_{i\in N}\sum_{S\ni i}\lambda _{S}R_{i}\left( a^{\ast }\right) \\
&=&\sum_{i\in N}R_{i}\left( a^{\ast }\right) \\
&=&v^{o}(N).
\end{eqnarray*}

Thus, $v^F$ is anti-balanced and its anti-core is non-empty.
\end{proof}

Combining our two main results, we obtain the following corollary.

\begin{corollary}
For all $P\in \mathcal{P}^-$, $\emptyset\neq\mathcal{A}\left(
v^{o}(\cdot,P)\right)\subseteq\mathcal{C}\left( {v}^{p}(\cdot,P)\right)$.
\end{corollary}

Thus, with very little structure on the problem other than negative
externalities, we are able to show the non-emptiness of the pessimistic core.

The guarantee of a non-empty anti-core does not carry over to problems with positive externalities, as illustrated in the following counterexamples.

\begin{example}
   Consider Example \ref{ex_irts}, and $v^L_{max}$. We need $y_1\leq 7$, $y_2\leq 0$ and $y_3\leq 0$, which are incompatible with $y_1+y_2+y_3=15$, and thus $\mathcal{A}(v^L_{max})=\emptyset$. The same is true for $\mathcal{A}(v^o)$.
\end{example}

\begin{example}
\cite{Atatrudeau2024} provide a variant of the queueing problem by supposing that agents must buy machines to queue on, and can buy as many machines as they want. The problem becomes one with positive feasibility externalities: by itself, a coalition can only buy its own machines and queue on them; if it joins others, it can still do so, but can also take advantage of unused time slots on their machines. Hence, in this case, choosing last corresponds to the optimistic approach, $v^{o}=v^{L}_{max}$. \cite{Atatrudeau2024} show that the core of the corresponding pessimistic game is sometimes empty, sometimes not. By Theorem \ref{main}, so is the anti-core of the optimistic game.
\end{example}



\section{Applications}
\label{sec:app}
In this section, we discuss several applications that exhibit feasibility externalities. We examine how these applications can be modeled within our framework, how optimistic and pessimistic approaches have been defined in each case, and whether our results allow to reinterpret existing results or obtain new ones. 

\subsection{Queueing problems}

We first examine more formally our example of queueing problems. Consider a set of agents $N$ that each have a job to be processed at one machine. The machine can process only one job at a time. Each agent $i\in N$ incurs waiting costs $w_{i}> 0$ per unit of time. The queueing problem determines both the order in which to serve agents and the corresponding monetary transfers they should receive (see \cite{chun16} for a survey on the queueing problem). 

In our framework, $R_{i}=-w_{i}r_{i}(\sigma)$ where $r_{i}(\sigma)$ is the rank of agent $i$ in the queue $\sigma$. Since a machine cannot serve more than a job at a given time, a queue $\sigma$ is feasible for any coalition if there are no $i,j\in N$ such that $r_{i}(\sigma)=r_{j}(\sigma)$. Then, $f_{S}(a_{F(N\setminus S)})$ is the feasible set for $S$ when the coalition $N\setminus S$ takes the first $|N\setminus S|$ positions in the queue. Hence, this problem exhibits negative feasibility externalities. Then, our results implies that choosing last corresponds to the pessimistic approach, $v^{L}_{min}=v^{p}$ and choosing first corresponds to the optimistic approach, $v^{F}=v^{o}$

These pessimistic and optimistic approaches have been defined independently in the literature. \cite{m03} built the optimistic game, using the assumption that a coalition is served before the players outside the coalition. The minimal transfer rule\footnote{The minimal transfer rule assigns to each agent a position in the queue and a monetary transfer. The monetary transfer is equal to half of their unit waiting cost multiplied by the number of agents in front of them in the queue subtracted by half of the sum of the unit waiting costs of the people behind them in the queue.}, $\phi^{min}$, is obtained by applying the Shapley value to $v^{o}$. Alternatively, \cite{chun06} assumes that a coalition is served after the non-coalitional members, obtaining the pessimistic game.\footnote{Independently, \cite{ks06} considered the same scenario as in \cite{chun06}. They introduced the associated game, the so-called tail game, and studied its core.} The maximal transfer rule\footnote{The maximal transfer rule assigns to each agent a position in the queue and a monetary transfer. The monetary transfer is equal to a half of the sum of the unit waiting costs of her predecessors minus a half of her unit waiting cost multiplied by the number of her followers.}, $\phi^{max}$, is obtained by applying the Shapley value to $v^{p}$.


It has been shown that for these queueing problems $v^o$ is concave and $v^p$ is convex, resulting in their Shapley values being respectively in the anti-core of $v^o$ and the core of $v^p$. We thus obtain the following results.

\begin{theorem}
For any queueing problem, we have $\phi^{min}\in \mathcal{A}(v^{o})\subseteq\mathcal{C}(v^p)$ and $\phi^{max}\in \mathcal{C}(v^p)$.
\end{theorem}

The results that $\{\phi^{min},\phi^{max}\}\in \mathcal{C}(v^p)$ come respectively from \cite{m03} and \cite{chun06}.
We obtain an additional justification for the minimal transfer rule, as it offers allocations that are below the optimistic bounds and above the pessimistic bounds. The maximal transfer rule offers allocations above the pessimistic bounds, but not always below the optimistic bounds.



\subsection{Minimum cost spanning tree problems}
\label{subsec:mcst}
We have a set of nodes consisting of $N_0\equiv N\cup \{0\}$, where $0$ is a special node called the source. Agents need to be connected to the source to obtain a good or a service. To each edge $(i,j)\in N_0\times N_0$ corresponds a cost $c_{ij}\geq 0$, with the assumption that $c_{ij}=c_{ji}$. These costs are fixed costs, paid once if an edge is used, regardless of how many agents use it. The problem is to connect all agents to the source at the cheapest cost. Given the assumptions above, among the optimal networks there always exists a spanning tree, hence the name of the problem. A minimum cost spanning tree (mcst) problem is $(N,c)$, where $c$ is the list of all edge costs. $c$ is also called a cost matrix.

In our framework, $R_i=-c_{p(i)i}$, where $p(i)\in N_0$ is the predecessor of $i$ in the unique path from 0 to $i$ in the spanning tree. The usual assumption is to suppose that a coalition $S$ cannot use edge $(i,j)$ if $i,j\in N_0\setminus S$. Then $f_S(\ominus)$ is the set of spanning trees rooted at $0$, while in $f_S(a_{N\setminus S})$, for any $a_{N\setminus S}$, we also treat agents $i\in N\setminus S$ such that $(a_{N\setminus S})_i\neq \ominus$ as additional sources. Thus, we obtain a problem with positive feasibility externalities. According to our results, we have $v^p=v^F$.

Most of the literature has considered that pessimistic game $v^p$, in which a coalition $S$ connects to the source first, before $N\setminus S$. An exception is \cite{bvp2007ijgt}, which considers the game in which coalition $S$ supposes that $N\setminus S$ has already connected to the source. In such a case, agents in $N\setminus S$, being connected, act as a source for $S$. Thus, how they are connected is irrelevant, and thus this game is equivalent, in our notation to both $v^L_{min}$ and $v^L_{max}$. 

While in general we cannot guarantee that for problems with positive externalities the optimistic game is the one in which a coalition chooses last, the discussion above shows that how the complement set $N\setminus S$ is connected is irrelevant, so there is no possibility for $S$ to obtain a better outcome by having some of its members select first. Thus, we also have that $v^{o}=v^{L}_{min}=v^{L}_{max}$. 

The literature has devoted considerable attention to the notion of irreducible cost matrix \citep{feltkamp1994irreducible, bvp2007jet}: since many edges are not used in any optimal spanning tree, we reduce the cost of these edges as much as possible, under the constraint that $v^p(N)$ does not change. There is a unique way to do so, and irreducible edge costs can be obtained as follows: take any optimal spanning tree, and for each pair of nodes $(i,j)\in N_0$, look at the (unique) path from one to the other, and assign to $(i,j)$ the most expensive edge on that path. We then obtain the irreducible cost matrix $\bar{c}$. Let $\bar{v}^p$ and $\bar{v}^o$ be the pessimistic and optimistic games obtained from the irreducible cost matrix, corresponding once again to choosing first/last.

\begin{theorem}
\citep{bvp2007ijgt}. For any mcst problem $(N,c)$ we have
\begin{enumerate}[\bfseries(i)]
    \item $\bar{v}^p$ and $\bar{v}^o$ are dual.
    \item $\bar{v}^o=v^o$.
\end{enumerate}
\end{theorem}
This leads us, using our results, to the following corollary.

\begin{corollary}
For any mcst problem $(N,c)$ we have $\mathcal{A}(v^o)=\mathcal{C}(\bar{v}^p)$.
\end{corollary}

This result is interesting for two reasons. First, $\mathcal{C}(\bar{v}^p)$ is called the irreducible core \citep{bird1976}, and has been shown to be uniquely characterized by additivity and monotonicity properties \citep{tijs2006, bvp2015}. Second, our equivalence with the anti-core of the optimistic game means that we do not need to go through the modification of the cost matrix into the irreducible matrix to obtain the irreducible core.

It is also worth noting that the allocation used to prove the non-emptiness of the anti-core of the optimistic game in Theorem \ref{main_non_empty} corresponds to the Bird allocation \citep{bird1976} in which each agent pays the cost of the edge connecting it to its nearest neighbor in its unique path to the source.

Finally, if we extend the model and suppose that we have a set of sources $\mathcal{S}$ with each agent $i$ needing to connect to a single source in $\mathcal{S}_i\subseteq\mathcal{S}$, then we might not have anymore that $v^o=v^L_{max}$, as it now matters to which sources the complement set is connected. It is easy to build an example where a coalition $S$ is better off having some of its members choose first to change the behavior of $N\setminus S$ so that they connect to sources that will also be used by the other members of $S$.

\subsection{River sharing problems}

Suppose a river described as a line with agents $i$ being upstream of agent $j$ if and only $i<j.$ There is an entry $e_{i}\geq 0$ of water at each location $i,$ and the water that flows at location $i$ can be consumed by agent $i$ or allowed to flow downstream. The benefit from water consumption for agent $i$ is given by a strictly increasing and strictly concave function $b_{i}.$ A water sharing problem is $(N,e,b),$ the set of players, the vector of water entries, and the collection of benefit functions \citep{AMBEC2002453}. The problem for the grand coalition is to maximize joint benefits, under the constraint imposed by the flows of water. If $x_{i}\geq 0$ is the consumption level of agent $i,$ the feasible set is constrained as follows: for any $i\in N$, that $\sum_{j\leq i}x_{j}\leq \sum_{j\leq i}e_{j}.$ For a coalition $S$, if the complement set is consuming any amount of water, the feasible set is reduced, and we thus have negative feasibility externalities. By our results, $v^o=v^F$ and $v^p=v^L_{min}$.


If coalition $S$ chooses first, it has access to all water entries in the river, subject to the physical constraints imposed by the river, i.e. an agent upstream of a location cannot consume the water entry at that location. Thus, we obtain that $%
v^{F}(S)=\max_{\left( x_{i}\right) _{i\in S}}\sum_{i\in S}b_{i}(x_{i})$
under the constraints that $\sum_{\substack{ j\leq i  \\ j\in S}}x_{j}\leq\sum_{j\leq i}e_{j}$ for all $i\in S.$  

If $S$ chooses last, then $N\setminus S$, given that its members are not satiable, have consumed as much water as they could. The exact maximizer is thus irrelevant, and we have $v^L\equiv v^L_{min}=v^L_{max}$. To define $v^L$ properly we need the following definition: a coalition is consecutive
if for any pair of agents in that coalition, adjoining agents are also in the coalition.
 Thus, we have
that $v^{L}(S)=0$ if $n\notin S$ and $v^{L}(S)=\max_{\left( x_{i}\right) _{i\in S^n}}\sum_{i\in S^n}b_{i}(x_{i})$ under the constraints that $\sum_{\substack{ j\leq i  \\ j\in S^n}}x_{j}\leq\sum_{\substack{j\leq i \\j\in S^n}}e_{j}$ for all $i\in S^n$
otherwise, where $S^{n}$ is the largest consecutive coalition in $S$ that contains $n.$ In words, if $i\in S$ is such that a member of $N\setminus S$ is downstream, then the water entries at $i$ and upstream have all been consumed by $N\setminus S$. Thus, the only group in $S$ that is able to consume is $S^n$, such that all its members are downstream of all members of $N\setminus S$. 

The coalitional functions proposed in the literature have been constructed from various doctrines used in international law. Under the unlimited territorial integrity (UTI) doctrine, an agent can consume any water that flows through its location. $v^{UTI}(S)$ is seen as an upper bound on the welfare of $S$, and it is easy to see that it corresponds to $v^F$. 

Under the absolute territorial sovereignty (ATS) doctrine, an agent has absolute rights over the water entering on its territory. For a single agent $i$, this implies that he should received at least $b_i(e_i)$. For larger coalitions, we suppose that an agent $i$ can transfer water to $j$ only if $j$ is its immediate downstream neighbour. Otherwise, the water is consumed by the agent(s) between $i$ and $j$.  Thus, let $\Gamma (S)$ be the coarsest partition of $S$ into consecutive coalitions. Then, $v^{ATS}(S)=\sum_{T\in \Gamma(S)}\max_{\left( x_{i}\right) _{i\in T}}\sum_{i\in T}b_{i}(x_{i})$ under the constraints that $\sum_{\substack{ j\leq i  \\ j\in T}}x_{j}\leq\sum_{\substack{ j\leq i  \\ j\in T}}e_{j}$ for all $i\in T$ and all $T\in\Gamma (S).$ 

Given the pessimistic constraints in the ATS version of the problem, $v^{ATS}(S)$ is seen as a lower bound on the welfare of $S$. But it is immediate that $v^{ATS}\geq v^L$, with, in particular, that $v^{ATS}(S)= v^L(S)$ if $S$ is a consecutive coalition containing $n$ and $v^{ATS}(S)\geq 0= v^L(S)$ if $S$ does not contain $n$. Thus, while pessimistic, $v^{ATS}$ is much less pessimistic than $v^L$.
We thus obtain that:
\[
v^p=v^L\leq v^{ATS}\leq v^{UTI}=v^F=v^o.
\]

We can see that in this model $v^p$ is a particularly pessimistic function, and only coalitions containing $n$ can guarantee a positive lower bound. It is no surprise that $v^{ATS}$ offers more interesting but higher lower bounds, showing that we must not always be interested in the most pessimistic lower bounds.   

\cite{AMBEC2002453} define the downstream incremental allocation rule as follows: $y_{i}^{DI}=v^{UTI}\left( \left\{ 1,\ldots,i\right\}\right)
-v^{UTI}\left( \left\{ 1,\ldots,i-1\right\} \right) =v^{ATS}\left(\left\{
1,\ldots,i\right\} \right) -v^{ATS}\left( \left\{ 1,\ldots,i-1\right\}\right)$. We have the following results.

\begin{theorem}
For all river sharing problem $(N,e,b),$ we have:

\begin{enumerate} 

\item[\bf{(i)}] \citep{AMBEC2002453}: $y^{DI}=\mathcal{A(}v^{UTI})\cap \mathcal{C(}%
v^{ATS})=\mathcal{A(}v^{o})\cap \mathcal{C(}%
v^{ATS});$ 

\item[\bf{(i)}] $y^{DI}\in \mathcal{C(}v^{p}).$
\end{enumerate}
\end{theorem}

Part (\textbf{ii}) is a simple consequence of our main theorem. \medskip%
\newline
Many extensions of the model have been considered, including to cases where
some agents can be satiated \citep{AMBECEHLERS2008} and cases with multiple
springs and bifurcations \citep{Khmelnitskaya2010}. See \cite{beal-river2012}
for a review.

\subsection{Applications with duality}
\label{sec:dual}

The applications and examples we have examined so far have all been such that studying the problem using lower and upper bounds gave different sets of allocations, and thus the perspective taken mattered. However, in some other cases, the two approaches are dual, and as seen in Proposition \ref{prop:dual}, it is unnecessary to study $v^o$ and $v^p$ separately if they are dual to each other. We discuss this duality and show that if, in addition, optimistic/pessimistic is equivalent to choosing first/last, then this duality can be easily spotted.

In our framework, $v^F$ and $v^L_{min}/v^L_{max}$ are dual if and only if for any coalition $S$, letting $S$ pick first and $N\setminus S$ react to that afterward always leads to an efficient outcome. In other words, an optimal outcome can always be obtained by sequential selfish optimizations by $S$ and $N\setminus S$. 

For all $S\subset N$ and all $a_S\in \mathbb{A}_S$, let $a_{N\setminus S}(a_S)$ be a best response of $N\setminus S$ to $S$ choosing $a_S$.

\begin{proposition}\label{dual}
For a problem $P\in \mathcal{P}$, we have that
\begin{enumerate}
\item[\bf{(i)}]  $v^{F}$ and $v_{min}^L$ are dual if and only if for all $S\subset N$ there exists $a_{S}\in \mu(\ominus_{N\setminus S})$ such that $a_S\in \min_{a'_S\in \mu(\ominus_{N\setminus S})}\max_{i\in N\setminus S}R_i(a_{N\setminus S}(a'_S))$ and  such that $\{ a_{S},a_{N\setminus S}(a_S)\}\in\arg\max\limits_{a_{N}\in f_{N}}\sum\limits_{i\in N}R_{i}(a_{N}).$
\item[\bf{(ii)}]  $v^{F}$ and $v_{max}^L$ are dual if and only if for all $S\subset N$ there exists $a_{S}\in \mu(\ominus_{N\setminus S})$ such that $a_S\in \max_{a'_S\in \mu(\ominus_{N\setminus S})}\max_{i\in N\setminus S}R_i(a_{N\setminus S}(a'_S))$ and  such that $\{ a_{S},a_{N\setminus S}(a_S)\}\in\arg\max\limits_{a_{N}\in f_{N}}\sum\limits_{i\in N}R_{i}(a_{N}).$
\end{enumerate}
\end{proposition}

\begin{corollary}[of Proposition \ref{prop:dual} and Theorem \ref{main}
\label{cor:dual}]{\ }
\begin{enumerate}
\item[\bf{(i)}]
For all $P\in \mathcal{P}^{-}$, if $v^F$ and $v^L_{min}$ are dual, then $\mathcal{A}\left( v^{o}\right) =\mathcal{C}\left(
v^{p}\right)$.
\item[\bf{(ii)}]
For all $P\in \mathcal{P}^{+}$, if $v^F$ and $v^L_{max}$ are dual, then $\mathcal{A}\left( v^{L}_{max}\right) =\mathcal{C}\left(
v^{p}\right)$.
\end{enumerate}    
\end{corollary}

Two important applications exhibiting duality are bankruptcy (claims) problems and airport problems.

The \emph{bankruptcy problem} deals with sharing an estate $E$ of a perfectly divisible resource among agents $N$ who have conflicting claims. That is, the sum of claims is larger than the estate: $\sum\limits_{i\in N}c_{i} > E$ where $c_{i}$ is the claim of agent $i$. \cite{oneill82} studied such problems from an economic point of view. He introduced an associated TU game to each bankruptcy problem and also defined the run-to-the-bank rule based on an average over all possible orders on agents' arrival. Within our framework, the action of agent $i$ is the amount that he takes from the estate, which is also his revenue. We thus have a negative feasibility externality problem.

As our results show, the optimistic approach corresponds to a bank-run situation, in which coalition $S$ arrives first and collects its combined claim or the endowment, whichever is smallest. The pessimistic approach has coalition $S$ arriving last, collecting what is left after the bank run of $N\setminus S$. The combination of the optimistic action of $S$ and the pessimistic action of $N\setminus S$ always leads to a full distribution of the endowment, and thus to an efficient outcome. Following Proposition \ref{dual}, the two games are dual.

The \emph{airport problem} introduced by \cite{lo73} aims to allocate the cost of a landing strip among users with varying runway length requirements. Every agent $i$ requires a length $l_{i}$ at the runway. It is assumed that the cost to build the runway is non-decreasing in its length. That is, for any two agents $i$ and $j$ such that $l_i<l_j$, $c(l_{i})\le c(l_{j})$. Within our framework, the action of agent $i$ is the segment of the runway that he builds, and his revenues is the cost of that segment.\footnote{His revenues are -$D$, with $D$ arbitrarily large, if agent $i$ does not have a runway of at least his desired length.} We thus have a problem with positive feasibility externalities.

According to our results, the pessimistic approach assumes that coalition $S$ arrives first to build its runway. The longest runway required by a member of the coalition will be built, which is $\max\limits_{i\in S}l_{i}$. Suppose next that coalition $N\setminus S$ picks last. Knowing that a runway of length $\max\limits_{i\in S}l_{i}$ has been built, it extends it, if needed, to a length of $\max\limits_{i\in N}l_{i}$. It is easy to see that this is the optimistic marginal contribution (by having some members of $N\setminus S$ choose first, it is impossible to get $S$ to build a longer runway). Given the inelastic demands of our agents, the length of the runway is efficient. Hence, the optimistic and pessimistic approaches are dual for airport problems.

We conclude this section by an illustration of the line between duality and non-duality. In a cooperative production problem, a set of agents share a production technology to produce some good(s). This joint production technology might exhibit increasing or decreasing returns to scale/scope.

As shown in Examples \ref{ex_drts} and \ref{ex_irts}, when the quantities consumed are endogenously determined (see, for instance, \citealp{MOULIN1990}; \citealp{rs93}; \citealp{FLEURBAEY1996}), the approach chosen matters, and a coalition $S$ choosing first and the complement $S$ might very well lead to an efficient quantity being produced and/or its distribution to agents being inefficient. Thus, we have no duality.

If we suppose that demands for the good(s) are exogenous (see, for instance, \citealp{ms92}; \citealp{m96}; \citealp{df98}), then producing these inelastic demands is efficient. It is easy to see that if the marginal cost of production is increasing when the number of units already produced, then choosing first is the best option and choosing last is the worst option. Given that demands are inelastic, we always obtain that the same (efficient) total quantity is produced by having $S$ and $N\setminus S$ sequentially choose. Thus, the optimistic and pessimistic games are dual.

\section{Extensions}
\label{sec:extension}

So far, we have considered what we call feasibility externalities, in which the actions taken affect the feasible sets of other agents. We now briefly discuss what happens when we have direct externalities, in which the actions taken by a group directly affects the revenues obtained by other agents. The additional difficulty in that case is what we have to determine who to credit for these direct externalities. 

We first notice that $v^{\alpha}$ and $v^{\beta}$ are now much more different than $v^F$ and $v^L_{min}/v^L_{max}$. They still, however, involve what we deem as non-credible threats. 

The links between first/last and best/worst scenarios are also weaker than without direct externalities. Suppose that both the direct and the feasibility externalities are negative, i.e. actions of $S$ limits the feasible set of $N\setminus S$ and have a direct negative impact on its revenue. Then, an optimistic viewpoint would be such that a coalition $S$ can completely ignore the externalities it generates on $N\setminus S$. A pessimistic viewpoint would force a coalition to internalize all negative externalities imposed on others.

Under the natural assumptions that externalities generated by $S$ are more important when combined with the actions of other agents we obtain the equivalent of Theorem \ref{v^T}i), and the best marginal contributions is when $S$ picks first. However, we lose Theorem \ref{v^T}ii) and it is not guaranteed that the worst marginal contribution is when $S$ picks last.

\begin{example}

To illustrate, we modify the queueing problem as follows. Suppose
that we have four agents with a job to process, but that the jobs are slightly
different, which requires an employee calibrating the machine to perform it
perfectly. For budgetary reasons, the employee can only perform a
calibration every two periods, meaning that it will provide an "average"
calibration for the two upcoming jobs that leaves the jobs not being
executed in a perfect manner. More precisely, the quality of the job depends
on who an agent is "matched" with, in the schedule. Suppose that if $i$ and $%
j$ are matched, the job of agent $i$ will be of a lower quality, resulting
in a cost of $d_{ij}\geq 0.$ Thus, if we use the schedule $i,j,k,l,$ there
is a waiting cost of $-w_{i}-2w_{j}-3w_{k}-4w_{l}$ and a matching cost of $%
-d_{ij}-d_{ji}-d_{kl}-d_{lk}.$

We show that in such a problem, even though we have negative externalities,
choosing last might not be when the marginal contribution is lowest. More
precisely, suppose that $w_{i}=1$ for all $i$ and suppose that $%
d_{13}=d_{31}=3,$ $d_{14}=d_{41}=d_{34}=d_{43}=4$, $d_{24}=d_{42}=5$ and $%
d_{ij}=0$ otherwise. In addition suppose that all agents have a deadline at
the fourth period, so that the waiting cost is arbitrarily large if
their job is not processed in the first four periods. In other words, $\mathbb{A}_i=\{1,2,3,4\}$ for all $i\in N$.

We consider $S=\left\{ 1,2\right\} .$ We are interested in its pessimistic
value, so we suppose that $S$ is held responsible for the externalities it
creates. We first examine what happens if $S$ chooses last. Then, $\left\{
3,4\right\} $ chooses first. If they take the first two spots, then they are
matched to each other. The waiting cost is 3 and the matching cost is 8$.$
Their other option is for one of them to take the first spot and the other
the third. Their waiting cost is 4, and since no match is created yet,
that's all they would pay. Thus that's what they select. Then, when $S$
comes in, it has to create the two matches and fully pay for the
externalities. Their option is to creates matches 1-3 and 2-4 or 2-3 and
1-4. In all cases, the waiting costs are 6. If they go with 1-3 and 2-4 they
create 16 in negative externalities. If they go with 2-3 and 1-4 they create
8. Thus, they pick 2-3 and 1-4, and $v_{\min }^{L}(S)=-14.$

We now consider $v_{\min }^{\left\{ 1\right\} \subseteq S}.$ Agent 1 has two
undominated strategies: either pick the first or the third position in the
queue, as picking second or fourth doesn't change the decision for $\left\{
3,4\right\} $ and only increases the waiting cost for agent 1. Suppose that
he picks the first position. Then, because agent 3 has a smaller disutility
than agent 4 if matched to agent 1, the only options for $\left\{
3,4\right\} $ are to pick positions 2 and 3 (waiting cost of 5 and matching
cost of 3) or positions 3 and 4 (waiting cost of 7 and matching cost of 8).
Thus, $\left\{ 3,4\right\} $ picks positions 2 and 3. Agent 2 then has no
choice but to pick position 4. Then, $S$ has a waiting cost of 5, and must
fully pay the externality created by the match of agents 2 and 4
(-10) and the disutility of agent 1 being matched to agent 3 (-3). That
results in a total value of $-18.$

If agent 1 picks the third position instead, then once again $\left\{
3,4\right\} $ has 2 options: match agent 3 with agent 1 (position 4) and put
agent 4 in position 1 (cost is 8) or pick positions 1 and 2 and match the
two members to each other (cost is 11). They thus pick the former, leaving
position 2 for agent 2. The value for $S$ is then composed of the waiting
cost ($-5)$, the full externality caused by the match between 2 and 4 $(-10)$
and the disutility of agent 1 being matched to agent 3 (-3). This also
results in a total value of $-18.$

Thus, $v_{\min }^{\left\{ 1\right\} \subseteq S}=-18<-14=v_{\min }^{L}(S).$
This occurs here because by having agent 1 commits first, it leads $%
N\setminus S$ to force the match between $2$ and $4,$ which is much worse
than the match between $2$ and $3.$

\end{example}
\medskip

However, by duality and by simplicity of calculation, we might want to consider $v^L_{min}$ as an approximation of the pessimistic value, especially since Theorem \ref{main} still applies, and $\mathcal{A}(v^F)$ and $\mathcal{C}(v^L_{min})$ provide subsets of $\mathcal{C}(v^p)$.   

\cite{trudeau-rosenthal2023}, studying a pipeline usage problem that generates negative (pollution) externalities, adopt a different position on how externalities should be credited. Using a polluters pay principle, they suppose that a coalition must compensate outsiders for the pollution generated by their usage. Thus, a coalition choosing first must internalize the damages caused to others when choosing their actions. In their particular setting, choosing first yields the largest marginal contributions, and choosing last yields the smallest ones.

In the case of positive externalities, the addition of direct externalities would still, under simple assumptions, leave $v^F$ as the pessimistic value. Since we already could not guarantee that choosing last was the optimistic value, this is further exacerbated by the presence of direct externalities, although, again, it might still be the case in some applications and a decent approximation in others.

\section{Concluding remarks}
\label{sec:conc}

In the presence of externalities, defining a coalitional value game requires making assumptions on the behavior of other agents. Given the multiple possibilities it is natural to define upper (optimistic) and lower (pessimistic) bounds on what a coalition should obtain. From a normative point of view, optimistic and pessimistic approaches should be analyzed with opposite solution concepts, as the former gives us upper bounds, while the latter provides lower bounds. In particular, if we use the core for games obtained from the pessimistic approach, we should use the anti-core for the optimistic approach. 

We have argued for the use of credible threats/gestures in determining these extreme bounds, in opposition to the classic definitions of $\alpha$ and $\beta$ games. In our settings with feasibility externalities, we have shown that there is a great benefit in carefully defining these optimistic and pessimistic approaches. Under negative externalities, the optimistic approach corresponds to a coalition choosing first, while the pessimistic approach corresponds to choosing last. In contrast, with positive externalities, while the pessimistic approach aligns with choosing first, the optimistic approach does not always correspond to choosing last.

In addition, when optimistic/pessimistic corresponds to first/last (or vice-versa), we obtain a powerful inclusion result: the anti-core of the optimistic game is always a subset of the core of the pessimistic game. Furthermore, we guarantee the non-emptiness of the anti-core in the presence of negative externalities, avoiding the need for restrictive structural assumptions.


Moreover, when the games defined based on whether a coalition chooses first or last are dual, it is redundant to study both games. Our results show that this duality is easy to spot: the greedy sequential choices of a coalition and its complement always lead to efficient outcomes.

Our main results are useful in multiple ways, particularly when optimistic/pessimistic approaches coincides with choosing first/last. First, it clearly indicates which bounds are easier to satisfy, and even if we believe that the core of the pessimistic game is more interesting, the anti-core of the optimistic game is a refinement as ensuring that nobody surpasses the optimistic bounds is more demanding than ensuring that nobody falls below the pessimistic bounds. In some applications, like in minimum cost spanning trees, this subset of allocations were shown to possess many interesting additional properties. Second, this also allows us an additional way to show that the core of the pessimistic game is non-empty. In fact, if externalities are negative, it is by going through the anti-core of the optimistic game that we can show that the core of the pessimistic game is always non-empty.

Through the applications, we learned that when we have positive feasibility externalities two conditions must be met for the optimistic function to not be obtained by choosing last. First, we need coalition $S$ to not only care if $N\setminus S$ is active or not, but to care about the precise actions they have taken. Second, we need coalition $S$ to be able to change the optimal behavior of $N\setminus S$ if some of its members choose first. In minimum cost spanning tree problems, the first condition is not met, as $S$ only cares if $N\setminus S$ is connected, but not how. In joint production problems with increasing returns to scale and inelastic demands, $S$ cannot change the behavior of $N\setminus S$, who will always choose to produce their inelastic demands. Thus, in both cases, the optimistic value is obtained when $S$ chooses last. The counterexamples we provide satisfy both conditions.

If we do not quite have that first/last is equivalent to optimistic/pessimistic, there are still benefits in using first/last as an approximation. It is computationally much simpler: calculating $v_{\min }^{T\subseteq S}$ involves
comparing $n!$ different combinations of strategies, so that calculating all 
$v_{\min }^{T\subseteq S}$ for all $T\subseteq S$ and all $S\subseteq N$ would require $%
2^{\left\vert S\right\vert }n!$ calculations. For the approximation using $v_{\min }^{L}(S)$ or $v_{\max }^{L}(S)$ we only require $n!$
calculations. This allows for the possibility of defining analytically the value function or the corresponding solution concepts like its Shapley value. In addition, the inclusion result is particularly useful to show non-emptiness or to define a subset of the pessimistic core when that set is too large.


Another possibility when evaluating the value created by a coalition in the presence of externalities is to take a weighted average of the optimistic and pessimistic values (see \citealp{banerjee2024}). Here, we provide a justification for this approach. Consider a game in which coalition $S$ comes in after $N\setminus S$ has chosen its optimal (selfish) actions. In opposition to what we have considered in this paper, suppose that coalition $S$ can pay $N\setminus S$ to ``undo'' its actions. A similar game, for a queueing game where a coalition must deal with an existing queue, was discussed in \cite{Atatrudeau2024}. If that offer is accepted, $S$ maximizes the welfare of the grand coalition and receives the value created, net of the payment to $N\setminus S$. Coalition $N\setminus S$ requires a payment of at least $v^F(N\setminus S)$ to accept to undo its actions. Coalition $S$ would be willing to pay as much as $v(N)-v^L(S)$ for this privilege. It is easy to see that in this game, if $S$ pays the maximum price, it obtains a net valuation of $v^L(S)$, and thus it is as if $S$ was picking last without this option. If it pays the minimum price, its net valuation is $v(N)-v^F(N\setminus S)$ and the resulting game is the dual of $v^F$. Thus, if we assume that the price paid is somewhere between these extremes, the resulting game is equivalent to a weighted average of $v^F$ and $v^L$, and given our results, possibly to a weighted average of $v^o$ and $v^p$.

\bibliographystyle{te}
\bibliography{opt-pes}
\end{document}